\documentclass[submission,copyright,creativecommons]{eptcs}
\providecommand{\event}{EXPRESS/SOS 2022} 
\providecommand{\thisvolume}[1]{this volume of EPTCS, Open Publishing Association} 

\usepackage{calculi}
\usepackage{microtype}
\usepackage{xcolor}
\usepackage{hyperref}

\usepackage[textsize = scriptsize, disable]{todonotes}

\usepackage{amsthm}
\newtheorem{theorem}{Theorem}
\newtheorem{lemma}{Lemma}

\newcommand{\coma}[1][\relax]{\ensuremath{\looparrowleft^{#1}}}
\newcommand{\subt}{\ensuremath{\leq}}

\hypersetup{
  unicode = true
}

\def\HOPI{{\normalfont HO}\PI}
\def\HOPSI{{\normalfont HO}\PSI}
\def\EPI{\ensuremath{\prescript{e}{}{\PI}}}

\def\ENV{\ensuremath{\Gamma}}
\def\BOM{\ensuremath{\vdash}}
\def\CH{\text{ch}}
\def\DR{\text{drop}}
\def\WRONG{\ensuremath{\textbf{WRONG}}}
\def\TJUDG{\ensuremath{\mathcal{J}}}
\def\TYPES{\ensuremath{\textbf{\textit{Types}}}}

\makeatletter
\ifengineTF{\pdftexengine}{
  \def\frame@args[#1]#2{\ensuremath{\mathcal{F}_{\kern-2pt #1}\kern-2pt\ARGS{#2}}}
}{}
\makeatother

\title{A Generic Type System for Higher-Order \PSI-calculi}
\author{%
\and
Alex Rønning Bendixen
\institute{Department of Computer Science\\ Aalborg University, Denmark}
\and
Bjarke Bredow Bojesen
\institute{Department of Computer Science\\ Aalborg University, Denmark}
\and
\and
Hans Hüttel 
\institute{Department of Computer Science\\ University of Copenhagen, Denmark}
\institute{Department of Computer Science\\ Aalborg University, Denmark}
\email{hans.huttel@di.ku.dk}
\and
Stian Lybech\thanks{The work by this author was partly supported by the Icelandic Research Fund Grant No.\@ 218202-05(1-3).} 
\institute{Department of Computer Science \\ Reykjavík University, Iceland}
\email{stian21@ru.is}
}
\def\titlerunning{A Generic Type System for Higher-Order \PSI-calculi}
\def\authorrunning{Bendixen, Bojesen, Hüttel and Lybech}

\let\asslparen=\llparenthesis
\let\assrparen=\rrparenthesis

\begin{document}
\ifengineTF{\xetexengine || \luatexengine}{
\def\llparenthesis{\asslparen}
\def\rrparenthesis{\assrparen}
}{}
\maketitle

\begin{abstract}
The Higher-Order \PSI-calculus framework (\HOPSI) is a generalisation of many first- and higher-order extensions of the \PI-calculus. 
It was proposed by Parrow et al.\@ who showed that higher-order calculi such as \HOPI{} and CHOCS can be expressed as \HOPSI-calculi.
In this paper we present a generic type system for \HOPSI-calculi which extends previous work by Hüttel on a generic type system for first-order \PSI-calculi.  
Our generic type system satisfies the usual property of subject reduction and can be instantiated to yield type systems for variants of \HOPI, including the type system for termination due to Demangeon et al.. 
Moreover, we derive a type system for the \RHO-calculus, a reflective higher-order calculus proposed by Meredith and Radestock.  
This establishes that our generic type system is richer than its predecessor, as the \RHO-calculus cannot be encoded in the \PI-calculus in a way that satisfies standard criteria of encodability.
\end{abstract}

\section{Introduction}
Process calculi are formalisms for modelling and reasoning about concurrent and distributed computations; a prominent example being the \PI-calculus of Milner et al.\@ \cite{PICALC, sangiorgi2003pi}, which models computation as communication between processes, by passing messages on named channels. 
Since its inception, a multitude of variants of the \PI-calculus have appeared; e.g.\@ D\PI{} \cite{DPICALC}, the calculus of explicit fusions \cite{gardner2000explicit}, the spi-calculus with correspondence assertions \cite{ABADI19991} and the \EPI-calculus \cite{CARBONEMAFFEIS}. 
These calculi are all \emph{first-order}, in the sense that only atomic channel names can be passed around, not processes themselves. 
Bengtson et al.\@ \cite{bengtson2009psi, bengtson2011psi} created \emph{\PSI-calculi} as a generalisation of these first-order variants and extensions, allowing a range of calculi, including all of the aforementioned, to be expressed as \emph{instances} of the \PSI-calculus framework through appropriate settings of a small number of parameters. 
However, there also exist \emph{higher-order} variants of the \PI-calculus, such as the Higher-Order \PI-calculus, \HOPI, \cite{HOPICALC, SANGIORGIPHD}, that also allow \emph{processes} to be sent across channels.   
Parrow et al.\@ \cite{parrow2014higher} have extended the \PSI-calculus framework with a construct for higher-order communication, creating the \emph{Higher-Order} \PSI-calculus, \HOPSI. 
Calculi such as \HOPI{} and CHOCS \cite{CHOCS} can now be represented as \HOPSI-instances, as well as every calculus that the `first-order' \PSI-calculus framework can represent. 

One of the techniques for reasoning about processes is that of \emph{type systems}.  
The first type system for a process calculus is due to Milner \cite{PICALC} and deals with the notion of correct usage of channels in the $\pi$-calculus: 
In a well-typed process only names of the correct type can be communicated.  
Pierce and Sangiorgi \cite{PICAPABLE} later described a type system that uses subtyping and capability tags to control the use of names as input or output channels; and also many of the aforementioned first-order extensions of the \PI-calculus have been given type systems to capture such properties as secrecy, authenticity and access control.

In \cite{huttel2011typed}, Hüttel noted that these type systems, despite arising in different settings, share certain characteristics:
The type judgments for processes $P$ are all of the form $\ENV \BOM P$ where $\ENV$ is a type environment recording the types of the free names in $P$, so processes are only classified as being either well-typed or not.  
On the other hand, terms $M$ are given a type $T$, so type judgements for terms are of the form $\ENV \BOM M : T$.  
Based on these shared characteristics, Hüttel then created a generic type system for the first-order \PSI-calculus framework, that generalises several of the type systems for the \PI-calculus and its variants.
This generic type system can similarly be instantiated through parameter settings to yield both well-known and new type systems for the calculi that are representable as first-order \PSI-calculi.  
An important advantage of this approach is that a general result of type system soundness can be formulated, which is then inherited by all instances of the type system.

There has been some other works on generic type systems, notably those of König \cite{konig2005}, Caires \cite{caires2007} and Igarashi and Kobayashi \cite{igarashi2004}. 
However, these are formulated for variants of the first-order \PI-calculus, which thus limits their applicability to languages that can be represented in the first-order paradigm. 
Stated otherwise, they exclude languages such as the aforementioned \EPI-calculus, which cannot be encoded into the first-order \PI-calculus, as shown in \cite{CARBONEMAFFEIS}, but which nevertheless can be given a type system using the generic approach of Hüttel, indicating that the latter is a more general framework for first-order calculi. 

However, the generic type system of Hüttel can only type \emph{first-order} calculi; it cannot be instantiated to yield type system for \emph{higher-order} calculi, such as \HOPI{} or CHOCS.
Both of these higher-order calculi can be encoded into the first-order \PI-calculus, as shown by Sangiorgi in \cite{HOPICALC}, and may therefore also be represented in just the first-order \PSI-calculus.
Not surprisingly, there is therefore little work on type systems for higher-order calculi, since these encodings allow us to disregard the higher-order behaviour and instead just type the first-order translations.
One exception is the type system for termination in variants of \HOPI{}, due to Demangeon et al.\@ \cite{demangeon2010550}.
As these authors argue, it may not always be desirable (or even possible) to type a higher-order language through a first-order representation, if the language contains features that are difficult (or impossible) to encode. 
For example, higher-order behaviour may alternatively be viewed as a special case of \emph{reflection}; i.e.\@ the ability of a program to turn code into data, modify or compute with it, and reinstantiate it as running code; and process mobility here appears as a special case where data (code/processes) are transmitted without modification.

This reflective capability is inherent in the Reflective Higher-Order (RHO or \RHO) calculus of Meredith and Radestock \cite{RHOCALC}, and this calculus cannot be uniformly encoded in the \PI-calculus, as shown in \cite{lybech2022encodabilitypaper}.
This calculus therefore gives us an example of a language that cannot easily be represented in the first-order paradigm, thus making it difficult (or impossible) to adapt any of the existing first-order type systems to this language.
Yet it \emph{can} be instantiated as a \HOPSI-calculus, as we shall show in the following.

The goal of the present paper is therefore to extend the aforementioned generic type system by Hüttel, to create a generic type system for the \HOPSI-calculus framework that will allow us to capture typability in the higher-order paradigm.
It allows us to identify what should be required of type systems for higher-order process calculi that are instances of the \HOPSI-calculus, and these requirements here take the form of a number of assumptions that must hold for each instance.
Like its predecessor, our generic type system also satisfies a general subject reduction property that is inherited by all instances.
We use this to formulate simple type systems for \HOPI, and we show that the type system for termination by Demangeon et al.\@ also can be captured as an instance of our type system.
Lastly, we show that our generic type system can be instantiated to yield a type system for the \RHO-calculus, which establishes that our type system is richer than the first-order type systems.
To our knowledge, no type system has hitherto been published for this calculus, so we regard this instance as a further contribution of the present paper.
A technical report with full proofs of most results is available in \cite{expresssos2022techreport}.

\begin{formatcalculus}[psi]  
\section{The higher-order \PSI-calculus}\label{sec:hopsi}
The Higher-Order \PSI-calculus extends the original \PSI-calculus \cite{bengtson2009psi} with primitives for higher-order communication, i.e.\@ process mobility.
In this section we first review the syntax of \HOPSI{} as given in \cite{parrow2014higher}, and then proceed to give a reduction semantics for the calculus.

\subsection{Syntax}
The Higher-Order \PSI-calculus is a general framework, which is intended to allow many different calculi to be obtained as instances, by setting a small number of parameters which takes the form of definitions of three (not necessarily disjoint) sets of \emph{terms, conditions} and \emph{assertions}.
To allow the framework to be as general and flexible as possible, the authors of \cite{bengtson2009psi, parrow2014higher} identify only a few restrictions that must be imposed on these sets: they must be \emph{nominal datatypes}. 
Informally, a \emph{nominal set}, in the sense of Gabbay and Pitts \cite{gabbay2002nominal}, is a set whose members can be affected by names being bound or swapped. 
If $a,b$ are names and $X$ is an element of a nominal set, then the \emph{transposition} of $a$ and $b$ on $X$, written $\SWAP{a}{b}{X}$, swaps all occurrences of $a$ for $b$ in $X$ and vice versa.
A function on a nominal set is \emph{equivariant}, if it is unaffected by name swapping; and a \emph{nominal datatype} is a nominal set together with a set of equivariant functions on it.
This requirement is very mild and allows e.g.\@ non-well-founded sets to be used in an instantiation.
The utility of this shall become apparent later, when we create a \HOPSI-calculus instance where the set of processes (which itself contains terms) is included in the set of terms.

Another important notion is that of support: if $X$ is an element of a nominal set, the \emph{support} of $X$, written $\SUPPORT{X}$, is the set of names that occur in $X$.  
Conversely, a name $a$ is \emph{fresh for $X$}, written $\FRESH{a}X$, if $a \notin \SUPPORT{X}$; and we extend this to sets of names $A$ such that $\FRESH{A}X$ if it is the case that $\forall a \in A. a \notin \SUPPORT{X}$.
This is pointwise extended to lists of elements $X_1, \ldots, X_n$, so we write $\FRESH{A}X_1, \ldots, X_n$ for $\FRESH{A}X_1 \land \ldots \land \FRESH{A}X_n$.

As mentioned above, any \PSI-calculus instance requires a specification of three nominal datatypes: the terms, conditions and assertions.  
The datatype of \emph{terms}, ranged over by $M,N \in \TERMS$, contains the terms that can be communicated and used as channels.  
These could be e.g.\@ single names, as in the monadic \PI-calculus, vectors of names as in \EPI{} and the polyadic \PI-calculus; or elements of a composite datatype (e.g.\@ the integers).  
The datatype of \emph{conditions}, ranged over by $\PHI \in \CONDITIONS$ contains the conditions that can be used in conditional process expressions.  
Finally, and importantly, we have the nominal datatype of \emph{assertions}, ranged over by $\PSI \in \ASSERTIONS$. 
Each of the datatypes \TERMS, \CONDITIONS{} and \ASSERTIONS{} must include an equivariant \emph{substitution function}, written $\ARGS{\cdot}\REPLACE{\VEC{a}}{\VEC{M}}$, substituting tuples of terms $\VEC{M}$ for tuples of names $\VEC{a}$ of equal arity.
It must be defined such that it satisfies the following \emph{substitution laws:} 
\begin{enumerate}
  \item If $\VEC{a} \subseteq \SUPPORT{X}$ and $b \in \SUPPORT[big]{\VEC{Y}}$ then $b \in \SUPPORT{ X\REPLACE{\VEC{a}}{\VEC{Y}} }$
  \item If $\FRESH{\VEC{u}}X,\VEC{v}$ then $ X\REPLACE{\VEC{v}}{\VEC{Y}} = \PAREN[big]{ \SWAP{\VEC{u}}{\VEC{v}}{X} }\REPLACE{\VEC{u}}{\VEC{Y}} $
\end{enumerate}

The requirements are quite general and should be satisfied by any ordinary definition of substitution: The first law states that names cannot be lost in substitution, i.e.\@ the names present in $\VEC{Y}$ must also be present when the substitution has been performed; whilst the second law states that substitution cannot be affected by transposition.

Since the calculus allows arbitrary terms to be used as channels, any \PSI-calculus instance requires a definition of two equivariant operators, \emph{channel equivalence} $\CHANEQ$ and \emph{assertion composition} $\COMPOSE$, a \emph{unit element} $\ASSUNIT$ of assertions, and an \emph{entailment relation} $\ENTAILS$, defined on the respective nominal datatypes and with the following signatures:

\begin{center}
\begin{math}
\begin{array}{ r @{~} l }
  \CHANEQ \;: \TERMS \times \TERMS \to \CONDITIONS          & \text{channel equivalence}   \\
  \COMPOSE : \ASSERTIONS \times \ASSERTIONS \to \ASSERTIONS & \text{assertion composition}
\end{array}\hspace*{0.5cm}
\begin{array}{ r @{~} l }
  \ASSUNIT \in \ASSERTIONS                                  & \text{assertion unit}        \\
  \ENTAILS \subseteq \ASSERTIONS \times \CONDITIONS         & \text{entailment relation}
\end{array}
\end{math}
\end{center}

\noindent We write the entailment relation as $\PSI \ENTAILS \PHI$ instead of $(\PSI, \PHI) \in\; \ENTAILS$ to denote that the condition $\PHI$ holds, given the assertions $\PSI$.
Note that comparison by channel equivalence $M_1 \CHANEQ M_2$ is itself a condition, which may or may not be entailed by some assertions $\PSI$, according to the definition of the entailment relation.

The set of \HOPSI-calculus \emph{processes} $\PROC$ is generated by the formation rules:
\begin{center}
\begin{syntax}[h]
  P \in \PROC
    \IS \NIL                              \NAME{Nil}
    \OR P_1 \PAR P_2                      \NAME{Parallel}
    \OR \OUTPUT{M}{N} . P                 \NAME{Output}
    \OR \INPUT{M}{\VEC{x}:\VEC{T}}{N} . P \NAME{Input}
    \ISOR \RUN{M}                         \NAME{Invocation}
    \OR \CASE \VEC{\PHI} : \VEC{P}        \NAME{Selection}
    \OR \NEW{x:T}P                        \NAME{Restriction}
    \OR \REPL{P}                          \NAME{Replication}
    \OR \ASSERT{\PSI}                     \NAME{Assertion}
\end{syntax}
\end{center}

\noindent where $\VEC{\PHI} : \VEC{P} \DEFSYM \PHI_1 : P_1 \CSOR \ldots \CSOR \PHI_n : P_n$.

Most of these constructs are similar to those of the \PI-calculus; the input and output prefixes generalise those of the \PI-calculus, since here both subject and object are \emph{terms} rather than just names.
Thus $\OUTPUT{M}{N}.P$ outputs the term $N$ on $M$ and continues as $P$, whilst $\INPUT{M}{\VEC{x}:\VEC{T}}{N}.P$ receives a term (e.g.\@ $K$) on $M$ that must match the pattern $N$.  
Here, $\VEC{x}$ is a list of pattern variables, binding into $N$ and $P$, that is used to extract subterms from $K$ that will then be substituted for the occurrences of $\VEC{x}$ within the continuation $P$.
Unlike the presentation in \cite{parrow2014higher}, we here use a typed version of the language: thus the types of the pattern variables are found in the list $\VEC{T}$ where $|\VEC{x}|=|\VEC{T}|$, and likewise, in the restriction $\NEW{x:T}P$, we annotate the name $x$ bound in $P$ with its type $T$.

The selection construct $\CASE \VEC{\PHI} : \VEC{P}$ is a shorthand for a list of cases and is to be understood as saying: If condition $\PHI_i$ is entailed by the assertions $\PSI$, we continue as $P_i$.
If more than one condition is entailed, the process is chosen non-deterministically.
This construct thus generalises the choice and matching operators of the \PI-calculus.

Higher-order communication is handled by representing processes as terms, thus allowing them to be communicated. 
We assume the existence of assertions of the form $M \HANDLE P$. 
By writing such an assertion, $M$ becomes a \emph{handle} of the process $P$, and we can then send $P$ by sending its handle.
Thence $M$ may be used to activate the process $P$, and for this we use the only construct that is new to the higher-order setting, the invocation construct $\RUN{M}$.
Note that the set of processes may itself be included in the set of terms, thus allowing assertions of the form $P \HANDLE P$ whereby a process becomes a handle for itself.

Lastly, an assertion $\ASSERT{\PSI}$ is said to be \emph{guarded}, if it occurs as a subterm of an input or output, and \emph{unguarded} otherwise.
The authors in \cite{parrow2014higher} impose the restriction that no assertion may occur unguarded in the processes in a conditional expression $\CASE \VEC{\PHI} : \VEC{P}$, nor in a replicated process $\REPL{P}$, nor in processes spawned by a $\RUN{M}$ operator. 
We say that processes conforming to this criterion are \emph{well-formed}, and we shall only consider well-formed processes in the following.

\subsection{Reduction semantics}\label{sec:psi-semantics}
Unlike previous presentations such as \cite{bengtson2011psi,parrow2014higher} we here use reduction semantics, as this will simplify our account of the generic type system.
As in other reduction semantics for process calculi, we introduce a notion of structural congruence, $\STRUCTEQ_S$, as the least congruence on process terms containing \ALPHA-equivalence, the commutative monoidal rules for parallel composition, and the rule for scope extrusion:
\begin{equation*}
 \RULENAME[s-scope]{S-Scope} ~ \NEW{x:T}P \PAR Q \STRUCTEQ_S \NEW{x:T}\PAREN{P \PAR Q} \quad\text{\normalfont if } \FRESH{x}Q 
\end{equation*}

We introduce a parametrised, asymmetric \emph{evaluation relation} $\cdot \ENVSEP \cdot \EVALTO \cdot$ to properly handle case expressions and unfolding of $\RUN{M}$ terms, both of which may depend on the assertions currently in effect.
It replaces the usual structural congruence rule in the reduction semantics to ensure that neither of these operations may be reversed by a reverse reading of the rules, whilst including $\STRUCTEQ_S$ for the other kinds of process rewrites where symmetry is unproblematic.
The \emph{reduction relation} $\cdot \ENVSEP \cdot \trans \cdot$, including the evaluation relation, is then given by the rules in figure~\ref{fig:reduction-semantics}, and reductions are thus on the form $\PSI \ENVSEP P \trans P'$, i.e.\@ relative to a global $\PSI$ containing the assertions currently in effect. 

\begin{figure}[t]
\begin{center}
\begin{semantics}
  \RULE[e-res][e-res](\FRESH{x}\PSI)
    { \PSI \ENVSEP P \EVALTO P' }
    { \PSI \ENVSEP \NEW{x:T}P \EVALTO \NEW{x:T}P' }

  \RULE[e-struct][e-struct]
    { P \STRUCTEQ_S P' }
    { \PSI \ENVSEP P \EVALTO P' }
\end{semantics}
\begin{semantics}
    \RULE[e-case][e-case]
    { \PSI \ENTAILS \PHI_i }
    { \PSI \ENVSEP \CASE \VEC{\PHI} : \VEC{P} \EVALTO P_i }

  \RULE[e-run][e-run]
    { \PSI \ENTAILS M \HANDLE P }
    { \PSI \ENVSEP \RUN{M} \EVALTO P }
\end{semantics}

\begin{semantics}
  \RULE[e-par][e-par](\FRESH{\FRAME[\nu]{Q}}{\PSI, \FRAME[\nu]{P}, P})
    { \PSI \COMPOSE \FRAME[\PSI]{Q} \ENVSEP P \EVALTO P' }
    { \PSI \ENVSEP P \PAR Q \EVALTO P' \PAR Q }
\end{semantics}
\begin{semantics}
  \RULE[e-rep][e-rep]
    { }
    { \PSI \ENVSEP \REPL{P} \EVALTO P \PAR \REPL{P} }
\end{semantics}

\begin{semantics}
  \RULE[r-com][r-com]
    { \PSI \ENTAILS M \CHANEQ K }
    { \PSI \ENVSEP \OUTPUT{M}{N\REPLACE{\VEC{x}}{\VEC{L}}}.P \PAR \INPUT{K}{\VEC{x}:\VEC{T}}{N}.Q \trans P \PAR Q\REPLACE{\VEC{x}}{\VEC{L}} }
\end{semantics}

\begin{semantics}
  \RULE[r-eval][r-eval]
    { \PSI \ENVSEP P \EVALTO Q \AND \PSI \ENVSEP Q \trans P' }
    { \PSI \ENVSEP P \trans P' }
\end{semantics}
\begin{semantics}
  \RULE[r-res][r-res](\FRESH{x}\PSI)
    { \PSI \ENVSEP P \trans P' }
    { \PSI \ENVSEP \NEW{x:T}P \trans \NEW{x:T}P' }
\end{semantics}

\begin{semantics}
  \RULE[r-par][r-par](\FRESH{\FRAME[\nu]{Q}}{\PSI, \FRAME[\nu]{P}, P})
    { \PSI \COMPOSE \FRAME[\PSI]{Q} \ENVSEP P \trans P' }
    { \PSI \ENVSEP P \PAR Q \trans P' \PAR Q }
\end{semantics}
\end{center}
\caption{Reduction semantics for the \HOPSI-calculus}
\label{fig:reduction-semantics}
\end{figure}
\end{formatcalculus}

New assertions $\ASSERT{\PSI}$ may also appear in the syntax and therefore become enabled during the evolution of the program.
These are collected by the \emph{frame function} $\FRAME[\PSI]{P}$ in the \nameref{r-par} and \nameref{e-par} rules; and likewise are any new names $\NEW{x:T}$ collected by $\FRAME[\nu]{P}$, used in the side conditions to ensure freshness of $x$ w.r.t.\@ $\PSI$ and the process in parallel composition.
The relevant clauses for $\FRAME[\PSI]{P}$ and $\FRAME[\nu]{P}$ are:
\begin{center}
\begin{math}
\begin{array}{r @{~} l}
  \FRAME[\PSI]{P \PAR Q}      & \DEFSYM \FRAME[\PSI]{P} \COMPOSE \FRAME[\PSI]{Q} \\
  \FRAME[\PSI]{\NEW{x:T}P}    & \DEFSYM \FRAME[\PSI]{P}                          \\
  \FRAME[\PSI]{\ASSERT{\PSI}} & \DEFSYM \PSI
\end{array}\hspace*{1cm}
\begin{array}{r @{~} l}
  \FRAME[\nu]{P \PAR Q}       & \DEFSYM \FRAME[\nu]{P} \UNION \FRAME[\nu]{Q}    \\
  \FRAME[\nu]{\NEW{x:T}P}     & \DEFSYM \SET{x} \UNION \FRAME[\nu]{P}           \\
                              &
\end{array}
\end{math}
\end{center}

\noindent and with all remaining clauses of the forms $\FRAME[\PSI]{P} \DEFSYM \ASSUNIT$ and $\FRAME[\nu]{P} \DEFSYM \EMPTYSET$ respectively.

\begin{formatcalculus}[psi]
\section{The generic type system}\label{sec:type-system}
The goal of our generic type system is to be able to instantiate it such that we obtain a sound type system for a given \HOPSI-calculus instance. 
As in other type systems, we need to describe when processes are well-typed, but since we in the \HOPSI-calculus also have terms, conditions and assertions, we shall therefore also need a way to decide when \emph{they} are well-typed.
However, since these nominal datatypes are parameters to the \HOPSI-calculus, we cannot specify a set of type rules for them, as we can with processes.
Instead, such rules must likewise be provided as parameters to create an instance of the generic type system, and these rules must then satisfy a number of requirements, here denoted \emph{instance assumptions}, which we shall need in the proof for subject reduction.
We describe them in detail below, in section~\ref{sec:parameters-and-assumptions}.

\subsection{Types and type judgements}
Types can contain names, and we assume that the set of types \TYPES{} is a nominal datatype ranged over by $T$; however, we do not allow substitution of terms for names \emph{inside} types.\footnote{As we shall see in our examples of instantiations of the type system, a type $T$ may itself contain a type environment $\Gamma$, which thus may contain names with type annotations.
By this requirement, we disallow that such names may be substituted for terms.}
Furthermore, we need the concept of a type environment $\ENV$ to record the types of free names; thus $\ENV$ is a partial function with finite support $\ENV : \NAMES \PARTIAL \TYPES$.
We can think of $\ENV$ as a set of tuples $\ENV \subseteq \NAMES \times \TYPES$ where $(x, T) \in \ENV$ if $\ENV(x) = T$, and we write $\ENV, x : T$ to denote the type environment $\ENV$ extended by the name $x$ with type $T$.

As usual, our type judgments will be relative to a type environment $\ENV$.
However, due to the presence of assertions which may affect the well-typedness of a process, term, condition, or indeed an assertion, our type judgments must also be relative to a global assertion $\PSI$.
As it may be composed with assertions appearing in a process, we shall therefore also need the notion of a specialisation preorder on assertions.  
We say that $\PSI_1 \leq \PSI_2$ if there exists a $\PSI$ such that $\PSI_2 = \PSI_1 \COMPOSE \PSI$, and $\supn(\PSI_1) \subseteq \supn(\PSI_2)$.

Given the above, type judgements for processes will be of the form $\ENV, \PSI \BOM P$.
As previously mentioned, the type rules for terms, assertions and conditions will depend on how these parameters are defined for a specific instance of the \HOPSI-calculus, and they must therefore be provided as part of the instantiation of the generic type system.
However, like type judgments for processes, they must also be relative to a type environment $\ENV$ and a global $\PSI$, so we require that they be of the form $\ENV, \PSI \BOM \TJUDG$, where $\TJUDG$ is defined by the formation rules:
\begin{center}
\begin{syntax}[h]
  \TJUDG \DEFSYM M : T \OR \PHI \OR \PSI
\end{syntax}
\end{center}

\subsection{Channel compatibility}
When we type an input or output prefix term, the type of the subject $M$ and the type of the object (the term transmitted on channel $M$) must be compatible w.r.t.\@ a \emph{compatibility predicate} $\coma$ that describes which types of values can be carried by channels of a given types.
Thus, $T_1 \looparrowleft T_2$ denotes that channels of type $T_1$ can carry terms of type $T_2$, and we require that the set of types be defined such that this holds.
Furthermore, we distinguish between output compatibility $\coma[+]$, and input compatibility $\coma[-]$, and we write $T_1 \coma T_2$ if both $T_1 \coma[+] T_2$ and $T_1 \coma[-] T_2$.

As an example, consider the channel types in the sorting system by Milner \cite{PICALC}. 
Here, a name has type $\CH(T)$, if it is a channel that can be used to transmit names of type $T$, so in that case we would therefore require that $\CH(T) \looparrowleft T$.

In our definition of compatibility, we assume given a subtype ordering $\subt$ on types. 
If $T_1 \subt T_2$, then a term of $T_1$ can be used wherever a term of type $T_2$ is needed.
Thus we require the usual subsumption rule for types, namely that a term of a given type $T_1$ can also be typed with a supertype $T_2$:
\begin{center}
\begin{semantics}
  \RULE[subsume][T-Sub] 
    { \ENV, \PSI \vdash M : T_1 \AND T_1 \subt T_2 }
    { \ENV, \PSI \vdash M : T_2 }
\end{semantics} 
\end{center}

\noindent The compatibility predicate for a type $T$ must further satisfy the following requirements w.r.t.\@ the subtyping relation:
\begin{enumerate}
\item If a channel type can carry two distinct types, then the types have to be related by the subtype ordering. 
  That is, if $d \in  \{+,-\}$, $T \coma[d] T_1$ and $T \coma[d] T_2$ with $T_1 \neq T_2$, then $T_1 \subt T_2$ or $T_2 \subt T_1$.

\item Output compatibility is \emph{contravariant}. 
  That is, if $T \coma[+] T_2$ and $T_1 \subt T_2$, then also $T \coma[+] T_1$. 
  This requirement mirrors that of \cite{PICAPABLE}. 
  If $T_1 \subt T_2$, then a term of type $T_1$ can be used where ever a term of type $T_2$ is needed, and a channel that outputs terms of the more general type $T_2$ can therefore be used, where ever a channel of the specialized type $T_1$ is required.

\item Input compatibility is \emph{covariant}. 
  That is, if $T \coma[-] T_1$ and $T_1 \subt T_2$, then also $T \coma[-] T_2$. 
  This requirement, too, mirrors that of \cite{PICAPABLE}. 
  Here, if $T_1 \subt T_2$, a channel that accepts terms of type $T_1$ can also be used to accept terms of type $T_2$.
\end{enumerate}

\subsection{Instance assumptions}\label{sec:parameters-and-assumptions}
In order to ensure soundness, we introduce a collection of assumptions, given in Figure~\ref{fig:instance-assumptions}, that must hold for an instance of the generic type system to be valid. 
They pertain to the type judgments $\ENV, \PSI \BOM \TJUDG$ for terms, conditions and assertions, which, as previously mentioned, we cannot specify in advance, but on which we must nevertheless impose certain restrictions to allow us to prove subject reduction for the generic type system.
Specifically, the assumptions will guarantee that the properties of weakening and strengthening and the substitution lemma will hold for any instance that satisfies these assumptions.

\begin{figure}
\begin{align*}
  \RULENAME[t-env-weak]{t-env-weak} ~ & 
    \ENV, \PSI \BOM \TJUDG \implies \ENV, x : T, \PSI \BOM \TJUDG\\
  \RULENAME[t-env-strength]{t-env-strength} ~ &
    \ENV, x : T, \PSI \BOM \TJUDG
    \land x \not\in \supn(\TJUDG)
    \implies \ENV, \PSI \BOM \TJUDG\\
  \RULENAME[t-comp-term]{t-comp-term} ~ &
    \ENV, \PSI \BOM M\REPLACE{\VEC{x}}{\VEC{L}} : F(\VEC{T})
    \implies \ENV, \PSI \BOM \VEC{L} : \VEC{T}\\
  \RULENAME[t-ass-weak]{t-ass-weak} ~ &
      \ENV, \PSI \BOM \TJUDG
      \land \PSI \le \PSI'
      \land \supn(\PSI') \subseteq \dom(\ENV)
      \implies \ENV, \PSI' \BOM \TJUDG\\
  \RULENAME[t-weak-chaneq]{t-weak-chaneq} ~ &
    \PSI \ENTAILS M_1 \CHANEQ M_2 \implies \PSI \COMPOSE \PSI' \ENTAILS M_1 \CHANEQ M_2\\
  \RULENAME[t-subs]{t-subs} ~ &
      \ENV, \PSI \BOM \VEC{L} : \VEC{T} \land
      \ENV, \VEC{x} : \VEC{T}, \PSI \BOM \TJUDG
    \implies   \ENV, \PSI \BOM \TJUDG\REPLACE{\VEC{x}}{\VEC{L}}\\
  \RULENAME[t-equal]{t-equal} ~ &
      \ENV, \PSI \BOM M : T \land
      \PSI \ENTAILS M \CHANEQ N
    \implies   \ENV, \PSI \BOM N : T\\
  \RULENAME[t-env-claus]{t-env-claus} ~ &
      \ENV, \PSI \BOM M : T \land T \curvearrowleft \ENV'
      \implies \dom(\ENV) \subseteq \dom(\ENV') \\
  \RULENAME[t-weak-ass-claus]{t-weak-ass-claus} ~ &
    \PSI \ENTAILS M \HANDLE P
    \land \ENV, \PSI \BOM M \HANDLE P
    \land \PSI \le \PSI'
    \land \supn(\PSI) \subseteq \ENV \implies \PSI' \ENTAILS M \HANDLE P\\
  \RULENAME[t-subs-run]{t-subs-run} ~ & \ENV, \PSI \BOM M : T \land T \curvearrowleft \ENV'
  \land \PSI \ENTAILS M\REPLACE{\VEC{x}}{\VEC{L}} \HANDLE P \implies \ENV', \PSI \BOM P
\end{align*}
\caption{Instance assumptions for the generic type system.}
\label{fig:instance-assumptions}
\end{figure}

Firstly, we require every instance of our generic type system to satisfy certain natural requirements about the use of type environments $\ENV$; these are similar to those of Hüttel in \cite{huttel2011typed}.
The assumptions \nameref{t-env-weak}, \nameref{t-env-strength}, \nameref{t-comp-term} and \nameref{t-ass-weak} are the usual requirements of weakening and strengthening; these must hold for type environments as well as for assertions. 
\nameref{t-weak-chaneq} tells us that channel equivalence is closed under weakening of assertions.
The assumptions \nameref{t-subs} and \nameref{t-equal} tell us that typability must be invariant under substitution and channel equivalence.

Other assumptions are particular to the higher-order setting and thus new.
Here, one particularly important question is \emph{which} type environment $\ENV$ a process $P$ should be typed in relation to, if $P$ is transmitted using the higher-order process mobility construct, with $M$ as a handle for $P$.
To solve this, we write $T \curvearrowleft \ENV$ to express that if $M$ is a handle for some process $P$ and has type $T$, then we can \emph{extract} the type environment $\ENV$ for typing $P$ from the type $T$ of the handle $M$.
This is thus another requirement we impose on how the set of types must be defined.
As a simple example, suppose that every type of a term would consist of a channel component and a run type component $(T, \ENV)$; then we could define the $\curvearrowleft$ relation to be $(T, \ENV) \curvearrowleft \ENV$.

The new assumptions are as follows:
\begin{itemize}
  \item The assumption \nameref{t-env-claus} tells us that that the type environment extracted from the type of a handle $M$ must mention at least the free names of $M$.
    
  \item The assumption \nameref{t-weak-ass-claus} is necessary to prove weakening of assertion environments; i.e.\@ by allowing unused assertions to be added. 
  It states that if $M$ is a handle for $P$, then $M$ must still remain a handle for the same process $P$ if the assertion environment is weakened.

  \item The assumption \nameref{t-subs-run} is needed to ensure that typability is preserved by substitution also in the higher-order case.
  It states that if a term $M$ becomes a handle for a new process $P$ after a substitution, then the new process must still be well-typed in the environment we obtain from $M$'s type $T$.
\end{itemize}

\subsection{Type rules for processes}
Unlike the aforementioned type rules for terms, conditions and assertions, the type rules for processes are common to every instance.
As in \cite{huttel2011typed}, we only consider type judgements that are \emph{well-formed}; that is, if $\SUPPORT{\PSI} \UNION \SUPPORT{P} \subseteq \dom(\ENV)$, so every name mentioned in the term or process in the judgement is bound in the type environment.
The rules are given in Figure~\ref{fig:typerules}; they are mostly similar to those of \cite{huttel2011typed}, except for the rule \nameref{t-run} used to type the $\RUN{M}$ construct, which is the only construct that is new to the higher-order setting.

\begin{figure}\centering
\begin{semantics}
  \RULE[t-in][t-in]{
    \begin{aligned}
                                          & T \looparrowleft T'\\
      \ENV, \PSI                     \BOM & M : T\\
      \ENV, \VEC{x} : \VEC{T}, \PSI  \BOM & N : T'\\
      \ENV, \VEC{x} : \VEC{T}, \PSI  \BOM & P\\
    \end{aligned}
  }{ \ENV, \PSI \BOM \INPUT{M}{\VEC{x}:\VEC{T}}{N}.P }
\end{semantics}
\begin{semantics}
  \RULE[t-run][t-run]{
    \begin{aligned}
                           & T \curvearrowleft \ENV'\\
      \ENV, \PSI  \BOM     & M : T\\
      \PSI        \ENTAILS & M \HANDLE P\\
      \ENV', \PSI \BOM     & P
    \end{aligned}
  }{ \ENV, \PSI \BOM \RUN{M} } 
\end{semantics}
\begin{semantics}
  \RULE[t-out][t-out] {
    \begin{aligned}
                      & T \looparrowleft T'\\
      \ENV, \PSI \BOM & M : T  \\
      \ENV, \PSI \BOM & N : T' \\
      \ENV, \PSI \BOM & P
    \end{aligned}
  }{ \ENV, \PSI \BOM \OUTPUT{M}{N}.P } 
\end{semantics}

\begin{semantics}
  \RULE[t-par][t-par](
   \begin{aligned}
      \FRESH{\FRAME[\nu]{P}}{\PSI, \FRAME[\nu]{Q}, Q}\\
      \FRESH{\FRAME[\nu]{Q}}{\PSI, \FRAME[\nu]{P}, P}
    \end{aligned}
  ){
    \ENV, \FRAME[\nu]{Q}, \PSI \COMPOSE \FRAME[\PSI]{Q} \BOM P \AND
    \ENV, \FRAME[\nu]{P}, \PSI \COMPOSE \FRAME[\PSI]{P} \BOM Q
  }{ \ENV, \PSI \BOM P \PAR Q }
\end{semantics}

\begin{semantics}
  \RULE[t-new][t-new](x \# \PSI)
    { \ENV, x : T, \PSI \BOM P }
    { \ENV, \PSI \BOM \NEW{x : T}P } 
\end{semantics}
\begin{semantics}
  \RULE[t-nil][t-nil]
    { }
    { \ENV, \PSI \BOM \NIL }
\end{semantics}
\begin{semantics}
  \RULE[t-repl][t-repl]
    { \ENV, \PSI \BOM P }
    { \ENV, \PSI \BOM \REPL{P} }
\end{semantics}

\begin{semantics}
  \RULE[t-case][t-case]
    { \ENV, \PSI \BOM \PHI_i \AND \ENV, \PSI \BOM P_i }
    { \ENV, \PSI \BOM \CASE{\VEC{\PHI} : \VEC{P}} } 
\end{semantics}
\begin{semantics}
  \RULE[t-assert][t-assert]
    { \ENV, \PSI \BOM \PSI' }
    { \ENV, \PSI \BOM \ASSERT{\PSI'} }
\end{semantics}
\caption{Type judgements for processes}
\label{fig:typerules}
\end{figure}

We shall comment on the rules in some detail:
In the rule for input, \nameref{t-in}, the subject $M$ must have type $T$, which must be compatible with the type $T'$ according to the aforementioned compatibility relation $\coma$.
The pattern $N$ must then have this type $T'$, given the assumptions that the variables $\VEC{x}$ have types $\VEC{T}$, and lastly, the continuation $P$ must be well-typed given these assumptions.
The output rule, \nameref{t-out} then mirrors the input rule as usual.
In both cases, the type judgment $\ENV, \PSI \BOM M : T$ appears in the premise, and as previously mentioned, the rules for this judgment must be provided as part of the instantiation.

In the rule \nameref{t-par} we require that for a parallel composition $P \PAR Q$ to be typable, $P$ and $Q$ must both be typable within type environments and assertions that add information extracted from the other component; thus we here overload the function $\FRAME[\nu]{\cdot}$ for $\NEW{x:T}P$ to mean $\FRAME[\nu]{\NEW{x:T}P} \DEFSYM x:T, \FRAME[\nu]{P}$.
This is a natural requirement, since $P$ can, among other things, mention handles for processes established in $Q$.
The side condition then asserts that all new names declared in $P$, using the $\NEW{x:T}$ construct, must be fresh for $\PSI$ \emph{and} both the free and new names occurring in $Q$, and vice versa for $Q$, similar to the side conditions for the \nameref{e-par} and \nameref{r-par} rules in the semantics.

Likewise in the rule \nameref{t-new}, we require that the new name $x$ must be fresh for $\PSI$, again mirroring the side conditions in the corresponding semantic rules \nameref{e-res} and \nameref{r-res}, and $P$ must then be well-typed given the assumption that $x$ has type $T$.

The rules for the nil process and replication, \nameref{t-nil} and \nameref{t-repl} are as usual, and the rule for $\CASE{\VEC{\PHI} : \VEC{P}}$ is also quite straightforward.
Here, we write $\ENV, \PSI \BOM \PHI_i$ and $\ENV, \PSI \BOM P_i$ to say that every condition $\PHI_i$ in the list of conditions $\VEC{\PHI}$, and every process $P_i$ in the list of processes $\VEC{P}$, must be well-typed w.r.t.\@ the same $\ENV$ and $\PSI$.
As in the cases for input and output above, the rules for concluding $\ENV, \PSI \BOM \PHI_i$ must be provided as part of the instantiation; and likewise for concluding $\ENV, \PSI \BOM \PSI'$, which appears in the premise of the \nameref{t-assert} rule.

Lastly, since a key motivation for the present type system is the ability to type higher-order behaviour, we must be able to describe what can happen when a handle $M \HANDLE P$ is released by a $\RUN{M}$.
This is handled by the rule \nameref{t-run}, which states that $\RUN{M}$ is well-typed for $\ENV$ and $\PSI$ if $M$ is a handle for $P$ in $\PSI$ and $P$ is well-typed in the environment $\ENV'$ extracted from $M$, using the aforementioned $\curvearrowleft$ relation. 

\section{Properties of the generic type system}
Type systems normally ensure two properties of well-typed programs: a \emph{subject reduction} property guarantees that a well-typed program remains well-typed under reduction; and a \emph{safety} property ensures that if a program is well-typed then a certain safety predicate holds.
The latter will depend on the particular instance of the type system and must therefore be shown individually, for each instance, but subject reduction can be shown for the generic type system.
We establish this through a series of lemmas, beginning with the usual results of \emph{weakening} and \emph{strengthening} of the type environment:
\begin{lemma}[Weakening and strengthening]
~
\begin{itemize}
\item  If $\ENV, \PSI \BOM P$ then $\ENV, x : T, \PSI \BOM P$ 
\item  If $\ENV, x : T, \PSI \BOM P$ and $\FRESH{x}{P, \PSI}$ then
  $\ENV, \PSI \BOM P$
  \end{itemize}
\end{lemma}

A similar result holds for assertions. 
Any process that is well-typed remains well-typed after a composition of any assertion in the assertion environment, so long as all names in the new assertion environment are in the support of the type environment: 
\begin{lemma}[Assertion environment weakening]\label{lemma:type-aew}
  If $\ENV, \PSI \BOM P$, $\names(\PSI') \subseteq \dom(\ENV)$ and $\PSI \leq \PSI'$ then $\ENV, \PSI' \BOM P$.
\end{lemma}

This lemma is necessitated by the syntax of the \HOPSI-calculus itself, which allows guarded assertions in continuations to become unguarded after a reduction. 
It is in the proof of this result that the instance assumptions \nameref{t-ass-weak}, \nameref{t-env-claus} and \nameref{t-weak-ass-claus} become important.

As we here use reduction semantics with an asymmetric evaluation relation to handle unfolding of \textbf{case} and \textbf{run} expressions, we shall also need two results that describe how frames can evolve during evaluation. 
The former, given in Lemma~\ref{lemma:frame-post-reduction} is used in the proof of subject reduction (Theorem~\ref{theorem:subject-reduction}) to find any new assertions that may have become composed onto the pre-existing assertion environment after a reduction.
The latter, given in  Lemma~\ref{lemma:frame-post-eval} states that the assertions in a process are unaltered by an evaluation: 
This is mainly ensured by the criterion for well-formed processes, asserting that all processes under replication or in a \textbf{case} expression, and all processes spawned by a $\RUN{M}$ operator, may not contain unguarded assertions.
The proof then establishes that the property of being assertion-guarded is preserved by the evaluation relation $\EVALTO$.
\begin{lemma}[Frame post reduction]\label{lemma:frame-post-reduction}
  If $\PSI \ENVSEP P \trans P'$ then $\FRAME[\PSI]{P} \leq \FRAME[\PSI]{P'}$
\end{lemma}

\begin{lemma}[Frame post evaluation]
  \label{lemma:frame-post-eval}
  If $\PSI \ENVSEP P \EVALTO P'$ then $\FRAME[\PSI]{P} = \FRAME[\PSI]{P'}$.
\end{lemma}

The above lemmas can now be used to prove that a well-typed process remains well-typed after an evaluation:
\begin{lemma}[Subject evaluation]\label{lemma:subject-eval}
  If~ $\ENV, \PSI \BOM P \land \PSI \ENVSEP P \EVALTO P'$ then $\ENV, \PSI \BOM P'$.
\end{lemma}

Lastly, we need a standard result of \emph{substitution}, which states that a well-typed process remains well-typed after a well-typed substitution. 
The proof of this lemma requires the instance assumptions \nameref{t-subs} and \nameref{t-subs-run}.
\begin{lemma}[Subject substitution]\label{lemma:subject-subs}
  If~ $\ENV, \VEC{x} : \VEC{T}, \PSI \BOM P$ and $\ENV, \PSI \BOM \VEC{L} : \VEC{T}$ then $\ENV, \PSI \BOM P\REPLACE{x}{\VEC{L}}$.
\end{lemma}

This, at last, gives us our main result:
\begin{theorem}[Subject reduction]\label{theorem:subject-reduction}
  If $\ENV, \PSI \BOM P \land \PSI \ENVSEP P \trans P'$ then $\ENV, \PSI \BOM P'$.
\end{theorem}
\begin{proof}[Outline] 
  Induction in the reduction rules. 
  In many of the cases, the instance assumptions are needed. 
  An example is that in the case of the \nameref{r-com} rule, the substitution assumption \nameref{t-subs} is needed to ensure that the substitution of the received message can be well-typed and the weakening assumptions \nameref{t-env-weak} and \nameref{t-ass-weak} are needed to ensure that the resulting process can be typed within the same type environment as before.
\end{proof}

The subject reduction theorem holds for all valid instances of the generic type system.  
This is all that we can guarantee in our generic setting, as a notion of \emph{safety} will also depend on a definition of runtime error, which will be specific to each instance. 
Safety must therefore be proved individually for each instance.
\end{formatcalculus}

\section{Instances of the generic type system}\label{sec:instances}
We now show how our generic type system can be applied to provide sound type systems for higher-order process calculi. 
We first consider type systems for a version of the \HOPI-calculus \cite{HOPICALC}, and then a type system for the \RHO-calculus \cite{RHOCALC} introduced by Meredith and Radestock.

\begin{formatcalculus}[psi]  
\subsection{The Higher-Order \PI-calculus}\label{sec:hopi}
Parrow et al.\@ \cite{parrow2014higher} give several examples of \HOPSI-instances with process mobility: for example, by including the set of processes $\PROC$ in $\TERMS$, a process $P$ may appear as the object of an output.
If for all $P \in \PROC . P \HANDLE P$ is entailed by all assertions, a language similar to Thomsen's Plain CHOCS \cite{thomsen1993plainchocs} is obtained, and by further allowing both names and processes to appear as objects of an output, we get a simplified version of Sangiorgi's \HOPI-calculus, similar to the one described in \cite{parrow2001introduction}.
We set the parameters for $\TERMS, \CONDITIONS$ and entailment thus:
\begin{center}
\begin{math}
\begin{array}{r @{~} l} 
  \TERMS      & \DEFSYM \NAMES \UNION \PROC  \\ 
  \CONDITIONS & \DEFSYM \SET[big]{ a \CHANEQ b | a, b \in \NAMES} \UNION \SET[big]{ P \HANDLE Q | P, Q \in \PROC } \UNION \SET{\TRUE} \\ 
  \ENTAILS    & \DEFSYM \SET[big]{ (\ASSUNIT, a \CHANEQ a) | a \in \NAMES } \UNION \SET{ (\ASSUNIT, P \HANDLE P) | P \in \PROC } \UNION \SET{(\ASSUNIT, \TRUE)}
\end{array}
\end{math}
\end{center}

\noindent and (initially) with $\ASSERTIONS \DEFSYM \SET{\EMPTYSET}$, $\COMPOSE \DEFSYM \UNION$ and $\ASSUNIT \DEFSYM \EMPTYSET$.
We also include the symbol $\TRUE$ in $\CONDITIONS$ to represent a condition that is entailed by all assertions, and use that for every condition in a $\CASE{\VEC{\PHI} : \VEC{P}}$ construct to obtain a representation of non-deterministic choice.
This parameter setting obviously allows unwanted processes such as
\begin{equation*}
  \OUTPUT{a}{P}.\NIL \PAR \INPUT{a}{x}{x} . \OUTPUT{x}{b} . \NIL \trans \OUTPUT{P}{b}.\NIL
\end{equation*}

\noindent where the process $P$ is substituted for the \emph{subject} $x$ in the output construct $\OUTPUT{x}{b}.\NIL$ after a reduction step.
However, we can now use our generic type system to create an instantiation that will disallow such possibilities. 
We define the types of terms as:
\begin{equation*}
  T \in \TYPES \DCLSYM \CH(T) \ORSYM \DR(\ENV)
\end{equation*}

\noindent The behaviour of channels and first-order variables is captured in the same manner as the simple sorting system for the \PI-calculus \cite{PICALC}. 
Process terms and higher-order variables will have the type $\DR(\ENV)$, where the processes are well-typed in $\ENV$.
Type errors can then be expressed as a simple error predicate, with
\begin{center}
\begin{tabular}{>{$} c <{$} @{\hspace*{2cm}} >{$} c <{$} }
  \dfrac
  { \ENV, \PSI \BOM M : \DR(\ENV') }
  { \ENV, \PSI \BOM \INPUT{M}{x}{x}.Q \trans \WRONG } 
  & 
  \dfrac
  { \ENV, \PSI \BOM M : \CH(T) }
  { \ENV, \PSI \BOM \RUN{M} \trans \WRONG }
\end{tabular}
\end{center}

\noindent as the most relevant rules.
We now define the instance parameters:
\begin{center}
\begin{semantics} 
  \RULE[t-con]
    { }
    { \ENV, \PSI \BOM \top }

  \RULE[t-cha]
    { }
    { \CH(T) \looparrowleft T }
\end{semantics}
\begin{semantics}
  \RULE[t-ass]
    { P : \DR(\ENV) \in \PSI' }
    { \ENV, \PSI \BOM \ASSERT{\PSI'} }

  \RULE[Term$_1$]
    { \ENV(x) = \CH(T) }
    { \ENV, \PSI \BOM x : \CH(T) }
\end{semantics}
\begin{semantics} 
  \RULE[t-end]
    { }
    { \DR(\ENV) \curvearrowleft \ENV }

  \RULE[Term$_2$]
    { P : \DR(\ENV') \in \PSI \quad \ENV', \PSI \BOM P}
    { \ENV, \PSI \BOM P : \DR(\ENV') }
\end{semantics}
\end{center}
\vskip-1em

Here we let the type environment in a drop type be the same type environment that is exposed to the processes, when it is defined as an object in an output prefix, i.e.\@ if we have $\ENV, \PSI \BOM \OUTPUT{a}{P}$, we want $\ENV, \PSI \BOM P : \DR(\ENV)$. 
In this way, when we run the process, we can recall the bound variables and their types at the time when the process was sent. 
To implement this, we shall use the (previously unused) assertions $\PSI$ as type environments for processes.
Thus we redefine $\ASSERTIONS$ as follows:
\begin{equation*}
  \ASSERTIONS \DEFSYM \POWERSET{\SET{P : T \ORSYM P \in \SETNAME{P} \land T \in \TYPES}}
\end{equation*}

\noindent We can now show safety for the type system instance:
\begin{theorem}
  If $\ENV, \PSI \BOM P$ then $P \not\trans \WRONG$.
\end{theorem}

\noindent The proof is by induction in the rules of $\ENV, \PSI \BOM P$.
Details are given in \cite{expresssos2022techreport}.
\end{formatcalculus}

\subsection{A type system for termination}
\newcommand{\HOpito}{\ensuremath{\textrm{HOpi}_2}}
\begin{formatcalculus}[pi]
  
We now turn our attention to an instance of the generic type system that captures a liveness property. 
Demangeon et al.\@ \cite{demangeon2010550} present a type system for checking termination in variants of the \HOPI-calculus: for any well-typed process $P$ we have that $P \trans* P' \not\trans$.
These authors study \HOpito, a higher-order process calculus in which only processes can be communicated. 
The syntax of \HOpito{} is given by the formation rules%
\begin{center}
\begin{syntax}[h]
  P \IS \NIL 
    \OR \INPUT{a}{X} . P 
    \OR \OUTPUT{a}{Q} . P 
    \OR P_1 \PAR P_2
    \OR \NEW{a:T}P 
    \OR X
\end{syntax}
\end{center}

In this type system, processes $P$ are typed with a type $n$, where $n$ is a natural number called the \emph{level} of $P$. 
Names $a$ have types of the form $\CH^k(\diamond)$, where $\diamond$ denotes the type of processes and $k$ is a natural number, the level of $a$. 
This is interpreted as saying that $a$ is only used to carry processes whose level $n$ is less than $k$. 
Type judgements are of the form $\ENV \BOM P : n$. 
The type rules, shown below, ensure that the level of processes that are sent on any channel $a$ will be strictly smaller than that of $a$.
\begin{flalign*} 
  \RULENAME[in]{In} \enspace
  & \dfrac{
    \begin{aligned}
      \ENV, X : (k-1) \BOM & P : n \\
      \ENV(a) = \CH^k(\diamond) 
    \end{aligned}
  }{
    \ENV \BOM a(X).P : n
  } &
  \RULENAME[out]{Out} \enspace
  &    \dfrac{\begin{aligned}
      \ENV \BOM Q : m \quad
      \ENV \BOM P : n \\
      \ENV(a) = \CH^k(\diamond) \quad m < k
    \end{aligned}}{\ENV \BOM \OUTPUT{a}{Q}.P : \max(k,n)}
 &
   \RULENAME[nil]{Nil} \enspace
  & \dfrac{
  }{
    \ENV \BOM \NIL : 0
  }
\end{flalign*}
\begin{flalign*}
  \RULENAME[new]{New} \enspace
  & \dfrac{
    \ENV, a : \CH^k(\diamond) \BOM P : n
  }{
    \ENV \BOM \NEW{a:T}P : n
  } &  
  \RULENAME[par]{Par} \enspace
  & \dfrac{
      \ENV \BOM P : m \quad
      \ENV \BOM Q : n
  }{
    \ENV \BOM P \PAR Q : \max(m,n)
  } &
  \RULENAME[var]{Var} \enspace
  & \dfrac{
    \ENV (X) = n
  }{
    \ENV  \BOM X : n
  }
\end{flalign*}

It is straightforward to represent the \HOpito\ calculus as an instance of the Higher-Order \PSI-calculus, using a variant of the parameter setting described in section~\ref{sec:hopi}. 
In order to represent the type system, we introduce assertions of the form
\begin{center}
\begin{syntax}[h]
  \PSI \IS n \OR n^- \OR n^+
\end{syntax}
\end{center}

\noindent We use assertions to indicate in which way a channel is to be used; an input use can only be typed in the presence of an assertion $n^{-}$ and output use must be used with an assertion $n^{+}$. 
We have that $n \COMPOSE n^- = n^- \COMPOSE n = n$; that $n \COMPOSE n^+ = n^+ \COMPOSE n = n$; and that $n_1 \COMPOSE n_2 = \max(n_1,n_2)$.
We distinguish explicitly between input uses ($\CH^{k}_{-} (\diamond)$) and output uses ($\CH^{k}_{+} (\diamond)$) of channels:\todo{R1: Explain why this is introduced. Is it related to the compatibility predicate?}
\begin{center}
\begin{syntax}[h]
  T \in \TYPES \IS n \OR \CH^{k}_{-} (\diamond) \OR \CH^{k}_{+} (\diamond)
\end{syntax}
\end{center}

\noindent and we let $\CH^n(\diamond) \curvearrowleft (n-1)$ and $\CH^n(\diamond) \curvearrowleft k$ whenever $k < n$.
Type judgements are of the form $\ENV,m \BOM M : T$ for terms and $\ENV,m \BOM P$ for processes. 
We represent the judgement $\ENV \BOM P : n$ as $\ENV,n \BOM P$.
The type rules for channels are thus:
\begin{flalign*} 
  \RULENAME[Term-in]{Ch-in} \enspace & \dfrac%
    { \ENV(a) = \CH^k_{-}(\diamond) }
    { \ENV,n^- \BOM a : \CH^k_{-}(\diamond) }
  & \RULENAME[Term-out]{Ch-out} \enspace & \dfrac%
    { \ENV(a) = \CH^k_{+}(\diamond) }
    { \ENV,n^+ \BOM a : \CH^k_{+}(\diamond) }
\end{flalign*}
\end{formatcalculus}

\subsection{The \RHO-calculus}
\begin{formatcalculus}[rho]
The Reflective Higher-Order calculus of Meredith and Radestock \cite{RHOCALC} is less well-known than e.g.\@ CHOCS and \HOPI, so we recall it in some detail.
Unlike other calculi, the \RHO-calculus does not assume an infinite set of names: instead, names and processes are both built from the same syntax, so names are structured terms, rather than atomic entities.
The syntax for both processes and names is given by the formation rules: 
\begin{center}
\begin{syntax}[h]
  P \IS \NIL 
    \OR P \PAR P
    \OR \LIFT{x}{P}
    \OR \INPUT{x}{y} . P
    \OR \DROP{x}
    \tabularnewline
x,y \IS \QUOTE{P}
\end{syntax}
\end{center}

\noindent where the syntax for names is simply $\QUOTE{P}$, pronounced \emph{quote $P$}.
Names can be passed around as in the \PI-calculus, as well as un-quoted (called \emph{drop}), and thus higher-order behaviour becomes an inherent property of the calculus, rather than just an extension on top of an already computationally complete language.

The parallel, and the input construct $\INPUT{x}{y} . P$, are similar to their \PI-calculus counterparts.
The \emph{lift} operation, $\LIFT{x}{P}$ is an output construct that quotes the process $P$, thereby creating the \emph{name} $\QUOTE{P}$, and sends it out on $x$; thus the calculus can generate new names at runtime without the need of a $\nu$-operator. 
The converse of lift is the \emph{drop} operation, $\DROP{x}$: it is a request to run the process within a name, by removing the quotes around it. 
This is not performed by a reduction, but rather by a form of substitution
\begin{equation*}
  \DROP{x}\REPLACE{\QUOTE{P}}{x'} = P \qquad\text{if } x \NAMEEQ x'
\end{equation*}

\noindent where the entire \emph{process} $\DROP{x}$ is replaced with the process $P$ found within the substituted name, similar to how process variables are replaced by processes in e.g.\@ \HOPI.
Notably, this means that if $x$ is a \emph{free} name, then $\DROP{x}$ will be a \emph{deadlock}, since $x$ can never be touched by a substitution at runtime.
Otherwise, substitution is the standard, capture-avoiding substitution of names for names, and note in particular that substitution does \emph{not} recur into processes under quotes; i.e.\@ $\QUOTE{P}\REPLACE{x}{y} = \QUOTE{P}$ if $y \not\NAMEEQ \QUOTE{P}$ regardless of whether the name $y$ exists somewhere within $\QUOTE{P}$.

The reduction semantics is given by the standard rules for parallel composition and structural congruence (as in e.g.\@ the \PI-calculus) plus the following rule for communication:
\begin{center}
\begin{tabular}{r @{} >{$} l <{$} }
  \RULENAME[rho:com]{\RHO-com} & \dfrac
    {x_1 \NAMEEQ x_2}
    {\INPUT{x_1}{y} . P \PAR \LIFT{x_2}{Q} \trans P\REPLACE{\QUOTE{Q}}{y} }
\end{tabular}
\end{center}

One subtlety of this calculus concerns the notion of structural congruence, $\STRUCTEQ$. 
It is the usual least congruence on processes, containing \ALPHA-equivalence, $\ALPHAEQ$, and the abelian monoid rules for parallel composition with $\NIL$ as the unit element.
However, with \emph{structured} terms as names, $\ALPHAEQ$ in turn requires a notion of \emph{name equivalence}, written $\NAMEEQ$, that is also used for comparing subjects in the \nameref{rho:com} rule above.
It is defined as the smallest equivalence relation on \emph{quoted} processes, closed forward under the rules: 
\begin{center}
\begin{tabular}{r @{} >{$} l <{$} @{\hspace*{1cm}} r @{} >{$} l <{$} }
  \RULENAME[rho:nameeq1]{\RHO-Nameeq$_1$} & \dfrac
    {P \STRUCTEQ Q}
    {\QUOTE{P} \NAMEEQ \QUOTE{Q}}         & 
  \RULENAME[rho:nameeq2]{\RHO-Nameeq$_2$} & \dfrac
    {}
    {\QUOTE{\DROP{x}} \NAMEEQ x}
\end{tabular}
\end{center}

This yields a mutual recursion between name equivalence, structural congruence and \ALPHA-equivalence, albeit one that always terminates as proved in \cite{RHOCALC}, because both the sets of names and processes are well-founded; their smallest elements being $\NIL$ (the inactive process) and $\QUOTE{\NIL}$ respectively.
\end{formatcalculus}

\subsubsection{Instantiation as a \PSI-calculus}
The \RHO-calculus is interesting in the present setting, because it cannot be encoded in the \PI-calculus in a way that satisfies a number of generally accepted criteria of encodability, similar to those of \cite{GORLA}.
This has been established by one of the authors in \cite{lybech2022encodabilitypaper}.

The key reason for this impossibility lies in the ability of the \RHO-calculus to generate new, \emph{free}, and hence observable, names at runtime, whilst this is not possible in the \PI-calculus; and, dually, its use of name equivalence, which will equate more names than strict syntactic equality.
However, the \RHO-calculus \emph{can} be represented in the \HOPSI-framework as follows.
We define
\begin{center}
\begin{math}
\begin{array}{r @{~} l}
  \TERMS      & \DEFSYM \NAMES \UNION \SET{ \QUOTE{P} | P \in \PROC[psi] } \UNION \SET{ \LQUOTE{P} | P \in \PROC[psi] } \\
  \CONDITIONS & \DEFSYM \SET[big]{ M \CHANEQ N | M, N \in \TERMS } \UNION \;\SET{ P_1 \STRUCTEQ P_2 | P_1, P_2 \in \PROC[psi]}        \\
              &  \UNION \;\SET{ M \HANDLE P | M \in \TERMS \land P \in \PROC[psi] } 
\end{array}
\end{math}
\end{center}

\noindent and (initially) with $\ASSERTIONS \DEFSYM \SET{\EMPTYSET}$, $\COMPOSE \DEFSYM \UNION$ and $\ASSUNIT \DEFSYM \EMPTYSET$ as before.
Note the two different kinds of terms: we use terms of the form $\QUOTE{P}$ to represent a \emph{statically} quoted name in the \RHO-calculus, which can never be dropped and never substituted into.
Conversely, we use $\LQUOTE{P}$ for the equivalent of the object of a $\LIFT{x}{P}$, which in the \RHO-calculus is a \emph{process} that therefore \emph{can} be substituted into, and which later may be dropped.
Furthermore, we shall assume that all \emph{bound names} are implemented as distinct atomic names $x \in \NAMES$; this is a trivial conversion, since their structure has no semantic meaning in the \RHO-calculus.
The encoding is then given by the translation:
\def\NTRANS#1{\PTRANS{\SETNAME{N}}[#1]}
\begin{center}
\begin{math}
\begin{array}{r @{~} l}
  \PTRANS[ \NIL ]                  & = \NIL                                                     \\
  \PTRANS[ P_1 \PAR P_2 ]          & = \PTRANS[ P_1 ] \PAR \PTRANS[ P_2 ]                       \\
  \PTRANS[ {\INPUT[rho]{n}{x}.P} ] & = \INPUT[psi] {\PTRANS[n]}{x}{\llift x\rlift} . \PTRANS[P] \\
  \PTRANS[ \LIFT{n}{P} ]           & = \OUTPUT[psi]{\PTRANS[n]}{ \LQUOTE{\PTRANS[P]}} . \NIL 
\end{array}\hspace*{1cm}
\begin{array}{r @{~} l}
  \PTRANS[ \DROP{x} ]              & = \RUN{ x }                                                \\
  \PTRANS[ \DROP{\QUOTE{P}} ]      & = \NIL                                                     \\
  \PTRANS[ \QUOTE{P} ]             & = \QUOTE{ \NTRANS{P} }                                     \\
  \PTRANS[ x ]                     & = x                                                        
\end{array}
\end{math}
\end{center}
\noindent where $\NTRANS{P}$ is similar to $\PTRANS[P]$ \emph{except} that $\NTRANS{\DROP{\QUOTE{P}}} = \RUN{\QUOTE{\NTRANS{P}}}$.

Note the two translations of drop for processes: the process $\DROP{\QUOTE{P}}$ has no reduction in the \RHO-calculus and is therefore behaviourally equivalent to $\NIL$; but its counterpart $\RUN{\QUOTE{P}}$ \emph{might} have a reduction, since $\RUN{M}$ is not evaluated eagerly in the \HOPSI-calculus.
For the purpose of preserving operational correspondence, we therefore translate the drop of a \emph{free name} $\QUOTE{P}$ as $\NIL$, and the drop of an atomic name $x$ as $\RUN{x}$, since atomic names are bound by construction.
However, we cannot do this within \emph{names}, since name equivalence is determined by the structure, rather than the behaviour of the process within quotes. 
Thus we use the second level translation $\NTRANS{P}$ for statically quoted names, since these can never be dropped.

Lastly, we shall define entailment such that it contains the rule $\PSI \ENTAILS \QUOTE{P} \HANDLE P$, making every term $\QUOTE{P}$ a handle for the process $P$ within, to mirror the duality of names and processes in the \RHO-calculus.
We furthermore include the following rules for entailment of \emph{channel equivalence} $\CHANEQ$, mirroring the rules \nameref{rho:nameeq1} and \nameref{rho:nameeq2} for concluding name equivalence:
\begin{center}
\begin{tabular}{r @{} >{$} l <{$} @{\hspace*{1cm}} r @{} >{$} l <{$} }
  \RULENAME[rho1:chaneq1]{Chaneq$_1$} & \dfrac
    {}
    { \PSI \ENTAILS \QUOTE{ \RUN{M} } \CHANEQ M }
    & 
  \RULENAME[rho1:chaneq2]{Chaneq$_2$} & \dfrac
    { \PSI \ENTAILS P_1 \STRUCTEQ P_2 }
    { \PSI \ENTAILS \QUOTE{P_1} \CHANEQ \QUOTE{P_2} }
\end{tabular}
\end{center}

\noindent including the symmetric and transitive closure of $\CHANEQ$.
We then let the entailment relation for conditions of \emph{structural congruence} $\STRUCTEQ$ be defined such that $\STRUCTEQ$ contains \ALPHA-equivalence; that $(\PROC_{/\STRUCTEQ}, \PAR, \NIL)$ is an abelian monoid; and containing the four congruence rules derived from the above translation:
\begin{center}
\begin{tabular}{r @{} >{$} l <{$} @{\hspace*{0.5cm}} r @{} >{$} l <{$} }
  \RULENAME{Par} & \dfrac
    { \PSI \ENTAILS P_1 \STRUCTEQ P_2 }
    { \PSI \ENTAILS P_1 \PAR R \STRUCTEQ P_2 \PAR R }
    &
  \RULENAME{In} & \dfrac
    { \PSI \ENTAILS M_1 \CHANEQ M_2 \qquad \PSI \ENTAILS P_1 \STRUCTEQ P_2 }
    { \PSI \ENTAILS \INPUT[psi]{M_1}{x_1}{\llift x_1\rlift} . P_1 \STRUCTEQ \INPUT[psi]{M_2}{x_2}{\llift x_2\rlift} . P_2 }
\end{tabular}
\begin{tabular}{r @{} >{$} l <{$} @{\hspace*{1.2cm}} r @{} >{$} l <{$} }
  \RULENAME{Run} & \dfrac
    { \PSI \ENTAILS M_1 \CHANEQ M_2 }
    { \PSI \ENTAILS \RUN{M_1} \STRUCTEQ \RUN{M_2} }
    &
  \RULENAME{Out} & \dfrac
    { \PSI \ENTAILS M_1 \CHANEQ M_2 \qquad \PSI \ENTAILS P_1 \STRUCTEQ P_2 }
    { \PSI \ENTAILS \OUTPUT[psi]{M_1}{\LQUOTE{P_1}} \STRUCTEQ \OUTPUT[psi]{M_2}{\LQUOTE{P_2}} }
\end{tabular}
\end{center}

This translation is sound and complete w.r.t.\@ operational correspondence up to a reasonable notion of behavioural equivalence $\BEHEQ$:
\begin{theorem}[Operational correspondence]
Let $\BEHEQ$ be a notion of behavioural equivalence for processes of the \HOPSI-instance of the \RHO-calculus, that includes at least structural congruence and the axiom $\RUN{\QUOTE{\PTRANS[P]}} \BEHEQ \PTRANS[P]$.
Then $P \trans P' \iff \PTRANS[P] \trans\BEHEQ \PTRANS[P']$.
\end{theorem}

\begin{proof}[Outline]
The proof requires a number of steps.
First we show that the translation preserves name equivalence; i.e.\@ $x_1 \NAMEEQ x_2 \iff \ASSUNIT \ENTAILS \PTRANS[x_1] \CHANEQ \PTRANS[x_2]$ by induction in the rules of name equivalence and structural congruence.
Then we show that substitution can be moved out of the translation; i.e.\@ $\PTRANS[P\sigma] \BEHEQ \PTRANS[P]\PTRANS[\sigma]$ and $\ASSUNIT \ENTAILS \PTRANS[(n)\sigma] \CHANEQ (\PTRANS[n])\PTRANS[\sigma]$, where $\sigma \DEFSYM \REPLACE[rho]{\QUOTE{Q}}{x}$ and $\PTRANS[\sigma] \DEFSYM \REPLACE[psi]{\PTRANS[x]}{\PTRANS[\QUOTE{Q}]}$, by induction in the clauses of the translation function. 
This step relies on our assumption about $\BEHEQ$.
Lastly we can show soundness and completeness w.r.t.\@ operational correspondence by induction in the two semantics.
In both cases, the interesting clauses are the communication rules, which require the aforementioned substitution and name equivalence preservation results.
Details are available in \cite{expresssos2022techreport}.
\end{proof}

\subsubsection{A type system for reflection}
Other higher-order calculi such as CHOCS and \HOPI{} can be encoded in the \PI-calculus and may thus be typable through translation, but as we noted above there cannot be such an encoding of the \RHO-calculus into the \PI-calculus.
Thus, we cannot hope to create a type system for the \RHO-calculus by adapting an existing first-order type system.
In fact, we are not aware of any type system for the \RHO-calculus, so we shall now create one by instantiating our generic type system.
We let types for names be of the form 
\begin{equation*}
  T \in \TYPES \DCLSYM \TUPLE{T, \ENV} \ORSYM \TUPLE{B, \ENV}
\end{equation*}

\noindent where $B$ is a basis type, and $\ENV$ is a type environment representing the possibility of executing the process within the name.
Furthermore we shall use assertions as type environments for processes
as we previously did with \HOPI, so we update the definition
accordingly.
\begin{align*}
  \ASSERTIONS & \DEFSYM \POWERSET{\SET{\QUOTE{P} : T \mid P \in \PROC[psi] \land T \in \TYPES }
    \UNION \SET{\LQUOTE{P} : T \mid P \in \PROC[psi] \land T \in \TYPES }}
\end{align*}

\noindent with assertion unit and composition as $\ASSUNIT \DEFSYM \EMPTYSET$ and $\COMPOSE \DEFSYM \UNION$ respectively.
Note that by construction $\forall x \in \NAMES[psi]. \FRESH{x}{\QUOTE{P}}$, so substitution can only occur in terms of the form $\LQUOTE{P} : T$.
We then append an assertion to the encoding of input and output:
\begin{align*}
  \PTRANS[ \LIFT{\QUOTE{R}}{P} ] & \DEFSYM \OUTPUT[psi]{\QUOTE{\PTRANS[R]}}{ \LQUOTE{\PTRANS[P]}} . \NIL \PAR \ASSERT{\SET{\QUOTE{\PTRANS[R]} : T, \LQUOTE{\PTRANS[P]} : T'}} \\
  \PTRANS[ {\INPUT[rho]{\QUOTE{R}}{x}.P} ] & \DEFSYM \INPUT[psi] {\QUOTE{\PTRANS[R]}}{x}{\llift x\rlift} . \PTRANS[P] \PAR \ASSERT{\SET{\QUOTE{\PTRANS[R]} : T}}
\end{align*}

Lastly, we also need to take the type information into account when concluding channel equivalence, to ensure that two terms with initially dissimilar types cannot become channel equivalent after a substitution.
Thus we redefine the entailment rule \nameref{rho1:chaneq2} as follows:
\begin{align*}
  \RULENAME{Chaneq$_2$} & \dfrac
    { \ENV, \PSI \ENTAILS P_1 \STRUCTEQ P_2 \qquad \ENV, \PSI \BOM \QUOTE{P_1} : T \iff \ENV, \PSI \BOM \QUOTE{P_2} : T }
    { \ENV, \PSI \ENTAILS \QUOTE{P_1} \CHANEQ \QUOTE{P_2} }
\end{align*}

\noindent Now we can instantiate the generic type system by defining the instance parameters:
\begin{center}
  \begin{tabular}{ @{} r @{} >{$} l <{$} @{\hspace*{0.3cm}} r @{} >{$} l <{$} @{\hspace*{0.3cm}} r @{} >{$} l <{$} @{} }
  \RULENAME[rho:t-term1]{Term-1} & \dfrac
    { \QUOTE{P} : \TUPLE{T, \ENV'} \in \PSI \qquad \ENV', \PSI \BOM P }
    { \ENV, \PSI \BOM \QUOTE{P} : \TUPLE{T, \ENV'} }
    &
  \RULENAME[rho:t-term2]{Term-2} & \dfrac
    { \LQUOTE{P} : \TUPLE{T, \ENV'} \in \PSI \qquad \ENV', \PSI \BOM P  }
    { \ENV, \PSI \BOM \LQUOTE{P} : \TUPLE{T, \ENV'} } 
  \end{tabular}
  \vskip5pt
  \begin{tabular}{ @{} r @{} >{$} l <{$} @{\hspace*{0.3cm}} r @{} >{$} l <{$} @{\hspace*{0.3cm}} r @{} >{$} l <{$} @{} @{\hspace*{0.3cm}} r @{} >{$} l <{$} }
  \RULENAME{t-ass} & \dfrac
    { P : T \in \PSI' \implies T \curvearrowleft \ENV }
    {\ENV, \PSI \BOM \ASSERT{\PSI'} }
    &
  \RULENAME{t-cha} & ~ \TUPLE{T, \ENV} \looparrowleft T 
    &
  \RULENAME{t-end} & ~\TUPLE{T, \ENV} \curvearrowleft \ENV
    &
  \RULENAME[rho:t-term3]{Term-3} & \dfrac
    { \ENV(x) = T }
    { \ENV, \PSI \BOM x : T }
  \end{tabular}
\end{center}

Note in particular the rules \nameref{rho:t-term1} and \nameref{rho:t-term2}: these rules say that the process within a term must be well-typed w.r.t.\@ the type environment in the second component of its type, and that the process-type pair must be represented in the assertion.\todo{R2: Example of ill-typed and well-typed processes?}

Since we include \nameref{rho1:chaneq2} in order to properly simulate the \RHO-calculus, all names that eventually become equal during reduction must have the same type.  
This amounts to requiring that the programmer must know in advance all the names that will be generated by the program during execution.  
We have yet to find a type system for the \RHO-calculus without this constraint.

\section{Conclusions and future work}\label{sec:conclusion}
We have presented a generic type system for higher-order \PSI-calculi, which extends a previous type system for first-order \PSI-calculi.
Like its predecessor, type judgements for processes are of the form $\ENV \BOM P$ and are given by a fixed set of rules.
Terms, assertions and conditions are assumed to form nominal datatypes, and only a few requirements on type rules are imposed.

The generic type system allows us to identify what should be required of type systems for higher-order process calculi that are instances of the \PSI-calculus; these requirements take the form of instance assumptions.
Thus it may also yield important insights into the general structure of type systems for higher-order calculi, and it may therefore also be taken as a starting point for developing more advanced type systems for any language that can be shown to be an instance of higher-order \PSI-calculi.

Our type system satisfies a general subject-reduction property and can be instantiated to yield type systems with a notion of channel safety for higher-order calculi such as CHOCS, \HOPI{} and also the \RHO-calculus. 
The latter in particular is interesting, as there is no valid encoding of the \RHO-calculus into the \PI-calculus, and thus we cannot capture higher-order typability in a purely first-order setting.  
This establishes that our generic type system is richer than first-order type systems.
However, typability in the \RHO-calculus comes at the cost of necessitating that we include type information directly in the definition of channel equivalence. 
This amounts to saying that the programmer must know (and specify) in advance the type of all names that will be generated during the course of program evaluation. 
We do not know whether it is possible to create other (non-trivial) type systems for the \RHO-calculus without such a restriction.

There are two important lines of future work in this direction:
In \cite{huttel2013resourcespsi}, Hüttel extends the generic type system to consider more general notions of subtyping and resource awareness, and in \cite{huttel2016sessionpsi} he also considers \emph{session types} for psi-calculi.
Both of these extensions are formulated for first-order \PSI-calculi only, and they would therefore be relevant to also consider in the higher-order setting.

\bibliographystyle{eptcs}
\bibliography{literature}

\begin{thebibliography}{10}
\providecommand{\bibitemdeclare}[2]{}
\providecommand{\surnamestart}{}
\providecommand{\surnameend}{}
\providecommand{\urlprefix}{Available at }
\providecommand{\url}[1]{\texttt{#1}}
\providecommand{\href}[2]{\texttt{#2}}
\providecommand{\urlalt}[2]{\href{#1}{#2}}
\providecommand{\doi}[1]{doi:\urlalt{https://doi.org/#1}{#1}}
\providecommand{\eprint}[1]{arXiv:\urlalt{https://arxiv.org/abs/#1}{#1}}
\providecommand{\bibinfo}[2]{#2}

\bibitemdeclare{article}{ABADI19991}
\bibitem{ABADI19991}
\bibinfo{author}{Martín \surnamestart Abadi\surnameend} \&
  \bibinfo{author}{Andrew~D. \surnamestart Gordon\surnameend}
  (\bibinfo{year}{1999}): \emph{\bibinfo{title}{A Calculus for Cryptographic
  Protocols: The Spi Calculus}}.
\newblock {\slshape \bibinfo{journal}{Information and Computation}}
  \bibinfo{volume}{148}(\bibinfo{number}{1}), pp. \bibinfo{pages}{1--70},
  \doi{10.1006/inco.1998.2740}.

\bibitemdeclare{techreport}{expresssos2022techreport}
\bibitem{expresssos2022techreport}
\bibinfo{author}{Alexander~R. \surnamestart Bendixen\surnameend},
  \bibinfo{author}{Bjarke~B. \surnamestart Bojesen\surnameend},
  \bibinfo{author}{Hans \surnamestart Hüttel\surnameend} \&
  \bibinfo{author}{Stian \surnamestart Lybech\surnameend}
  (\bibinfo{year}{2022}): \emph{\bibinfo{title}{Typing Reflection in
  Higher-Order Psi-calculi}}.
\newblock \bibinfo{type}{Technical Report}, \bibinfo{institution}{Department of
  Computer Science, Aalborg University}.
\newblock
  \urlprefix\url{http://icetcs.ru.is/stian/2022/hopsitypes2022techreport.pdf}.

\bibitemdeclare{inproceedings}{bengtson2009psi}
\bibitem{bengtson2009psi}
\bibinfo{author}{Jesper \surnamestart Bengtson\surnameend},
  \bibinfo{author}{Magnus \surnamestart Johansson\surnameend},
  \bibinfo{author}{Joachim \surnamestart Parrow\surnameend} \&
  \bibinfo{author}{Bj{\"o}rn \surnamestart Victor\surnameend}
  (\bibinfo{year}{2009}): \emph{\bibinfo{title}{Psi-calculi: Mobile processes,
  nominal data, and logic}}.
\newblock In: {\slshape \bibinfo{booktitle}{2009 24th Annual IEEE Symposium on
  Logic In Computer Science}}, \bibinfo{organization}{IEEE}, pp.
  \bibinfo{pages}{39--48}, \doi{10.1016/S1571-0661(05)80361-5}.

\bibitemdeclare{article}{bengtson2011psi}
\bibitem{bengtson2011psi}
\bibinfo{author}{Jesper \surnamestart Bengtson\surnameend},
  \bibinfo{author}{Magnus \surnamestart Johansson\surnameend},
  \bibinfo{author}{Joachim \surnamestart Parrow\surnameend} \&
  \bibinfo{author}{Björn \surnamestart Victor\surnameend}
  (\bibinfo{year}{2011}): \emph{\bibinfo{title}{{Psi-calculi: a framework for
  mobile processes with nominal data and logic}}}.
\newblock {\slshape \bibinfo{journal}{{Logical Methods in Computer Science}}}
  \bibinfo{volume}{{Volume 7, Issue 1}}, \doi{10.2168/LMCS-7(1:11)2011}.
\newblock \urlprefix\url{https://lmcs.episciences.org/696}.

\bibitemdeclare{inproceedings}{caires2007}
\bibitem{caires2007}
\bibinfo{author}{Luís \surnamestart Caires\surnameend} (\bibinfo{year}{2007}):
  \emph{\bibinfo{title}{Logical Semantics of Types for Concurrency}}.
\newblock In \bibinfo{editor}{Till \surnamestart Mossakowski\surnameend},
  \bibinfo{editor}{Ugo \surnamestart Montanari\surnameend} \&
  \bibinfo{editor}{Magne \surnamestart Haveraaen\surnameend}, editors:
  {\slshape \bibinfo{booktitle}{Algebra and Coalgebra in Computer Science}},
  \bibinfo{publisher}{Springer Berlin Heidelberg}, \bibinfo{address}{Berlin,
  Heidelberg}, pp. \bibinfo{pages}{16--35}, \doi{10.1006/inco.1994.1093}.

\bibitemdeclare{article}{CARBONEMAFFEIS}
\bibitem{CARBONEMAFFEIS}
\bibinfo{author}{Marco \surnamestart Carbone\surnameend} \&
  \bibinfo{author}{Sergio \surnamestart Maffeis\surnameend}
  (\bibinfo{year}{2003}): \emph{\bibinfo{title}{On the Expressive Power of
  Polyadic Synchronisation in Pi-Calculus}}.
\newblock {\slshape \bibinfo{journal}{Nordic Journal of Computing}}
  \bibinfo{volume}{10}(\bibinfo{number}{2}), pp. \bibinfo{pages}{70--98},
  \doi{10.1016/S1571-0661(05)80361-5}.

\bibitemdeclare{article}{demangeon2010550}
\bibitem{demangeon2010550}
\bibinfo{author}{Romain \surnamestart Demangeon\surnameend},
  \bibinfo{author}{Daniel \surnamestart Hirschkoff\surnameend} \&
  \bibinfo{author}{Davide \surnamestart Sangiorgi\surnameend}
  (\bibinfo{year}{2010}): \emph{\bibinfo{title}{Termination in higher-order
  concurrent calculi}}.
\newblock {\slshape \bibinfo{journal}{The Journal of Logic and Algebraic
  Programming}} \bibinfo{volume}{79}(\bibinfo{number}{7}), pp.
  \bibinfo{pages}{550--577}, \doi{10.1016/j.jlap.2010.07.007}.
\newblock \bibinfo{note}{The 20th Nordic Workshop on Programming Theory (NWPT
  2008)}.

\bibitemdeclare{article}{gabbay2002nominal}
\bibitem{gabbay2002nominal}
\bibinfo{author}{Murdoch \surnamestart Gabbay\surnameend} \&
  \bibinfo{author}{Andrew \surnamestart Pitts\surnameend}
  (\bibinfo{year}{2002}): \emph{\bibinfo{title}{A New Approach to Abstract
  Syntax with Variable Binding}}.
\newblock {\slshape \bibinfo{journal}{Formal Asp. Comput.}}
  \bibinfo{volume}{13}, pp. \bibinfo{pages}{341--363},
  \doi{10.1007/s001650200016}.

\bibitemdeclare{inproceedings}{gardner2000explicit}
\bibitem{gardner2000explicit}
\bibinfo{author}{Philippa \surnamestart Gardner\surnameend} \&
  \bibinfo{author}{Lucian \surnamestart Wischik\surnameend}
  (\bibinfo{year}{2000}): \emph{\bibinfo{title}{Explicit fusions}}.
\newblock In: {\slshape \bibinfo{booktitle}{International Symposium on
  Mathematical Foundations of Computer Science}},
  \bibinfo{organization}{Springer}, pp. \bibinfo{pages}{373--382},
  \doi{10.1007/3-540-44612-5\_33}.

\bibitemdeclare{article}{GORLA}
\bibitem{GORLA}
\bibinfo{author}{Daniele \surnamestart Gorla\surnameend}
  (\bibinfo{year}{2010}): \emph{\bibinfo{title}{Towards a unified approach to
  encodability and separation results for process calculi}}.
\newblock {\slshape \bibinfo{journal}{Information and Computation}}
  \bibinfo{volume}{208}(\bibinfo{number}{9}), pp. \bibinfo{pages}{1031--1053},
  \doi{10.1016/j.ic.2010.05.002}.

\bibitemdeclare{article}{DPICALC}
\bibitem{DPICALC}
\bibinfo{author}{Matthew \surnamestart Hennessy\surnameend} \&
  \bibinfo{author}{James \surnamestart Riely\surnameend}
  (\bibinfo{year}{2002}): \emph{\bibinfo{title}{Resource Access Control in
  Systems of Mobile Agents}}.
\newblock {\slshape \bibinfo{journal}{Information and Computation}}
  \bibinfo{volume}{173}(\bibinfo{number}{1}), pp. \bibinfo{pages}{82--120},
  \doi{10.1006/inco.2001.3089}.

\bibitemdeclare{inproceedings}{huttel2011typed}
\bibitem{huttel2011typed}
\bibinfo{author}{Hans \surnamestart Hüttel\surnameend} (\bibinfo{year}{2011}):
  \emph{\bibinfo{title}{Typed $\psi$-calculi}}.
\newblock In: {\slshape \bibinfo{booktitle}{International Conference on
  Concurrency Theory}}, \bibinfo{organization}{Springer}, pp.
  \bibinfo{pages}{265--279}, \doi{10.1007/978-3-642-23217-6\_18}.

\bibitemdeclare{inproceedings}{huttel2013resourcespsi}
\bibitem{huttel2013resourcespsi}
\bibinfo{author}{Hans \surnamestart Hüttel\surnameend} (\bibinfo{year}{2014}):
  \emph{\bibinfo{title}{Types for Resources in $\psi$-calculi}}.
\newblock In \bibinfo{editor}{Martín \surnamestart Abadi\surnameend} \&
  \bibinfo{editor}{Alberto \surnamestart Lluch~Lafuente\surnameend}, editors:
  {\slshape \bibinfo{booktitle}{Trustworthy Global Computing}},
  \bibinfo{publisher}{Springer International Publishing},
  \bibinfo{address}{Cham}, pp. \bibinfo{pages}{83--102},
  \doi{10.1007/978-3-319-05119-2\_6}.

\bibitemdeclare{inproceedings}{huttel2016sessionpsi}
\bibitem{huttel2016sessionpsi}
\bibinfo{author}{Hans \surnamestart Hüttel\surnameend} (\bibinfo{year}{2016}):
  \emph{\bibinfo{title}{Binary Session Types for Psi-Calculi}}.
\newblock In \bibinfo{editor}{Atsushi \surnamestart Igarashi\surnameend},
  editor: {\slshape \bibinfo{booktitle}{Programming Languages and Systems}},
  \bibinfo{publisher}{Springer International Publishing},
  \bibinfo{address}{Cham}, pp. \bibinfo{pages}{96--115},
  \doi{10.1007/978-3-319-47958-3\_6}.

\bibitemdeclare{article}{igarashi2004}
\bibitem{igarashi2004}
\bibinfo{author}{Atsushi \surnamestart Igarashi\surnameend} \&
  \bibinfo{author}{Naoki \surnamestart Kobayashi\surnameend}
  (\bibinfo{year}{2004}): \emph{\bibinfo{title}{A generic type system for the
  Pi-calculus}}.
\newblock {\slshape \bibinfo{journal}{Theoretical Computer Science}}
  \bibinfo{volume}{311}(\bibinfo{number}{1}), pp. \bibinfo{pages}{121--163},
  \doi{10.1016/S0304-3975(03)00325-6}.

\bibitemdeclare{article}{konig2005}
\bibitem{konig2005}
\bibinfo{author}{Barbara \surnamestart König\surnameend}
  (\bibinfo{year}{2005}): \emph{\bibinfo{title}{Analysing
  input/output-capabilities of mobile processes with a generic type system}}.
\newblock {\slshape \bibinfo{journal}{The Journal of Logic and Algebraic
  Programming}} \bibinfo{volume}{63}(\bibinfo{number}{1}), pp.
  \bibinfo{pages}{35--58}, \doi{10.1016/j.jlap.2004.01.004}.
\newblock \bibinfo{note}{Special issue on The pi-calculus}.

\bibitemdeclare{inproceedings}{lybech2022encodabilitypaper}
\bibitem{lybech2022encodabilitypaper}
\bibinfo{author}{Stian \surnamestart Lybech\surnameend} (\bibinfo{year}{2022}):
  \emph{\bibinfo{title}{Encodability and Separation for a Reflective and
  Higher-Order Language}}.
\newblock \bibinfo{series}{\thisvolume{2}}.

\bibitemdeclare{article}{RHOCALC}
\bibitem{RHOCALC}
\bibinfo{author}{L.G. \surnamestart Meredith\surnameend} \&
  \bibinfo{author}{Matthias \surnamestart Radestock\surnameend}
  (\bibinfo{year}{2005}): \emph{\bibinfo{title}{A Reflective Higher-order
  Calculus}}.
\newblock {\slshape \bibinfo{journal}{Electronic Notes in Theoretical Computer
  Science}} \bibinfo{volume}{141}(\bibinfo{number}{5}), pp. \bibinfo{pages}{49
  -- 67}, \doi{10.1016/j.entcs.2005.05.016}.
\newblock \bibinfo{note}{Proceedings of the Workshop on the Foundations of
  Interactive Computation (FInCo 2005)}.

\bibitemdeclare{incollection}{PICALC}
\bibitem{PICALC}
\bibinfo{author}{Robin \surnamestart Milner\surnameend} (\bibinfo{year}{1993}):
  \emph{\bibinfo{title}{The Polyadic $\pi$-Calculus: a Tutorial}}.
\newblock In: {\slshape \bibinfo{booktitle}{Logic and Algebra of
  Specification}}, \bibinfo{publisher}{Springer Berlin Heidelberg}, pp.
  \bibinfo{pages}{203--246}, \doi{10.1007/978-3-642-58041-3\_6}.

\bibitemdeclare{incollection}{parrow2001introduction}
\bibitem{parrow2001introduction}
\bibinfo{author}{Joachim \surnamestart Parrow\surnameend}
  (\bibinfo{year}{2001}): \emph{\bibinfo{title}{An introduction to the
  $\pi$-calculus}}.
\newblock In: {\slshape \bibinfo{booktitle}{Handbook of Process Algebra}},
  \bibinfo{publisher}{Elsevier}, pp. \bibinfo{pages}{479--543},
  \doi{10.1016/B978-044482830-9/50026-6}.

\bibitemdeclare{article}{parrow2014higher}
\bibitem{parrow2014higher}
\bibinfo{author}{Joachim \surnamestart Parrow\surnameend},
  \bibinfo{author}{Johannes \surnamestart Borgstr{\"o}m\surnameend},
  \bibinfo{author}{Palle \surnamestart Raabjerg\surnameend} \&
  \bibinfo{author}{Johannes \surnamestart {\AA}man~Pohjola\surnameend}
  (\bibinfo{year}{2014}): \emph{\bibinfo{title}{Higher-order psi-calculi}}.
\newblock {\slshape \bibinfo{journal}{Mathematical Structures in Computer
  Science}} \bibinfo{volume}{24}(\bibinfo{number}{2}),
  \doi{10.1017/S0960129513000170}.

\bibitemdeclare{inproceedings}{PICAPABLE}
\bibitem{PICAPABLE}
\bibinfo{author}{Benjamin \surnamestart Pierce\surnameend} \&
  \bibinfo{author}{Davide \surnamestart Sangiorgi\surnameend}
  (\bibinfo{year}{1993}): \emph{\bibinfo{title}{Typing and subtyping for mobile
  processes}}.
\newblock In: {\slshape \bibinfo{booktitle}{[1993] Proceedings Eighth Annual
  IEEE Symposium on Logic in Computer Science}}, \bibinfo{organization}{IEEE},
  pp. \bibinfo{pages}{376--385}, \doi{10.1109/LICS.1993.287570}.

\bibitemdeclare{phdthesis}{SANGIORGIPHD}
\bibitem{SANGIORGIPHD}
\bibinfo{author}{Davide \surnamestart Sangiorgi\surnameend}
  (\bibinfo{year}{1993}): \emph{\bibinfo{title}{Expressing mobility in process
  algebras: first-order and higher-order paradigms}}.
\newblock Ph.D. thesis, \bibinfo{school}{University of Edinburgh}.
\newblock \urlprefix\url{http://hdl.handle.net/1842/6569}.

\bibitemdeclare{inproceedings}{HOPICALC}
\bibitem{HOPICALC}
\bibinfo{author}{Davide \surnamestart Sangiorgi\surnameend}
  (\bibinfo{year}{1993}): \emph{\bibinfo{title}{From $\pi$-calculus to
  higher-order $\pi$-calculus --- and back}}.
\newblock In \bibinfo{editor}{M.~C. \surnamestart Gaudel\surnameend} \&
  \bibinfo{editor}{J.~P. \surnamestart Jouannaud\surnameend}, editors:
  {\slshape \bibinfo{booktitle}{TAPSOFT'93: Theory and Practice of Software
  Development}}, \bibinfo{publisher}{Springer Berlin Heidelberg},
  \bibinfo{address}{Berlin, Heidelberg}, pp. \bibinfo{pages}{151--166},
  \doi{10.1007/3-540-56610-4\_62}.

\bibitemdeclare{book}{sangiorgi2003pi}
\bibitem{sangiorgi2003pi}
\bibinfo{author}{Davide \surnamestart Sangiorgi\surnameend} \&
  \bibinfo{author}{David \surnamestart Walker\surnameend}
  (\bibinfo{year}{2003}): \emph{\bibinfo{title}{The pi-calculus: a Theory of
  Mobile Processes}}.
\newblock \bibinfo{publisher}{Cambridge university press}.

\bibitemdeclare{inproceedings}{CHOCS}
\bibitem{CHOCS}
\bibinfo{author}{Bent \surnamestart Thomsen\surnameend} (\bibinfo{year}{1989}):
  \emph{\bibinfo{title}{A Calculus of Higher Order Communicating Systems}}.
\newblock In: {\slshape \bibinfo{booktitle}{Proceedings of the 16th {ACM}
  {SIGPLAN}-{SIGACT} symposium on Principles of programming languages - {POPL}'
  89}}, \bibinfo{series}{POPL'89}, \bibinfo{publisher}{{ACM} Press},
  \bibinfo{address}{New York, NY, USA}, pp. \bibinfo{pages}{143--154},
  \doi{10.1145/75277.75290}.

\bibitemdeclare{article}{thomsen1993plainchocs}
\bibitem{thomsen1993plainchocs}
\bibinfo{author}{Bent \surnamestart Thomsen\surnameend} (\bibinfo{year}{1993}):
  \emph{\bibinfo{title}{Plain {CHOCS} A Second Generation Calculus for Higher
  Order Processes}}.
\newblock {\slshape \bibinfo{journal}{Acta Inf.}}
  \bibinfo{volume}{30}(\bibinfo{number}{1}), p. \bibinfo{pages}{1–59},
  \doi{10.1007/BF01200262}.

\end{thebibliography}
\end{document}